\documentclass[letterpaper, 10 pt, conference]{ieeeconf}  

\IEEEoverridecommandlockouts  
\usepackage[utf8]{inputenc}
\usepackage{centernot}
\usepackage{xcolor}
\usepackage{graphicx}
\usepackage{graphics}
\usepackage{subcaption,comment}
\usepackage{amsmath}
\usepackage[version=4]{mhchem}
\usepackage{siunitx}
\usepackage{longtable,tabularx}
\usepackage{algorithm}
\usepackage{algorithmicx}
\usepackage{algpseudocode}
\usepackage{amsmath}
\usepackage{amssymb}
\usepackage{proof}
\usepackage{epsfig}
\usepackage{calc}

\usepackage{amsthm}
\usepackage{accents}
\usepackage{cite}
\usepackage{hyperref}

\newtheorem{problem}{Problem}

\newtheorem{remark}{Remark}
\newtheorem{lemma}{Lemma}
\newtheorem{definition}{Definition}
\newtheorem{assumption}{Assumption}
\newtheorem{theorem}{Theorem}
\newtheorem{example}{Example}
\newcommand{\bx}{{\boldsymbol{x}}}
\newcommand{\bu}{{\boldsymbol{u}}}
\newcommand{\by}{{\boldsymbol{y}}}
\newcommand{\bz}{{\boldsymbol{z}}}

\newcommand{\bv}{{\boldsymbol{v}}}
\newcommand{\br}{{\boldsymbol{r}}}
\newcommand{\hbd}{{\boldsymbol{\hat{d}}}}
\newcommand{\dbd}{{\boldsymbol{\dot{d}}}}
\newcommand{\ed}{{\boldsymbol{e}_d}}

\newcommand{\bd}{{\boldsymbol{d}}}
\newcommand{\argmin}{\operatorname{argmin}}    

\usepackage{xcolor}

\title{\LARGE \bf
Disturbance Observer-based Robust Integral Control Barrier Functions for Nonlinear Systems with High Relative Degree
}

\author{Vrushabh Zinage \and Rohan Chandra \and Efstathios Bakolas
\thanks{This research has been supported  in part by NSF award CMMI-1937957.}
\thanks{Vrushabh Zinage (graduate student) and Efstathios Bakolas (Associate Professor) are with the Department of Aerospace Engineering and Engineering Mechanics,
The University of Texas at Austin, Texas 78712-1221, USA, 
Rohan Chandra is with the Department of Computer Science, The University of Texas at Austin, Austin, Texas 78712, USA, 
{\tt\small \{vrushabh.zinage,rchandra,bakolas\}@utexas.edu}}
}

\begin{document}

\bibliographystyle{IEEEtran} 

\maketitle
\thispagestyle{empty}
\pagestyle{empty}

\begin{abstract}
In this paper, we consider the problem of safe control synthesis of general controlled nonlinear systems in the presence of bounded additive disturbances. Towards this aim, we first construct a governing augmented state space model consisting of the equations of motion of the original system, the integral control law and the nonlinear disturbance observer. Next, we propose the concept of Disturbance Observer based Integral Control Barrier Functions (DO-ICBFs) which we utilize to synthesize safe control inputs. The characterization of the safe controller is obtained after modifying the governing integral control law with an additive auxiliary control input which is computed via the solution of a quadratic problem. In contrast to prior methods in the relevant literature which can be unnecessarily cautious due to their reliance on the worst case disturbance estimates, our DO-ICBF based controller uses the available control effort frugally by leveraging the disturbance estimates computed by the disturbance observer. By construction, the proposed DO-ICBF based controller can ensure state and input constraint satisfaction at all times. Further, we propose Higher Order DO-ICBFs that extend our proposed method to nonlinear systems with higher relative degree with respect to the auxiliary control input. Finally, numerical simulations are provided to validate our proposed approach.

\end{abstract}

\section{Introduction}
Control Barrier Functions (CBF) have been proven to be effective in guaranteeing safety for control-affine systems and applied in many real-world applications such as aerospace \cite{breeden2023robust_aerospace}, robotics \cite{chandra2023decentralized_robotics_1,zinage2023neural_cbf}, multi-agent systems \cite{jankovic2022multi_agent_systems} etc. The inherent characteristics of CBFs make them suitable as safety filters for nominal stabilizing controllers that might not have been originally designed with safety in mind.

Traditional CBF's guarantee safety for control-affine systems by iteratively solving a quadratic program (QP). Furthermore, solving the QP iteratively can lead to recursive feasibility which would be detrimental in safety-critical applications. This makes the application of CBF slightly restrictive to some input-constrained real-world applications such as hypersonic pursuit \cite{lee2022feedback_hypersonic}, flapping birds \cite{chirarattananon2017dynamics_flapping_birds,zhang2022bio_flapping_bird_2}, soft robots \cite{george2020first_soft_robots,rus2015design_soft_robots_2}, and robotic systems with Ackermann steering geometry \cite{bascetta2016kinematic_ackermann} where the governing dynamics is non-control affine in nature. To address this issue, recently \cite{ames2020integral_cbf} proposed Integral Control Barrier functions (ICBF's) which is able to encode the state as well as input constraints directly onto a scalar function. Furthermore, by leveraging these ICBF's, they were able to modify the original integral control law via an additive auxiliary control input which was designed to guaranteed safety. However, when there is additional uncertainty in the model or external disturbances, the safety guarantees provided by ICBF's can be compromised or change. In addition, the analysis of the paper is only restricted to systems of relative degree one where the relative degree is with respect to the auxiliary control input.

To handle external disturbances in the model, robust versions of traditional CBFs have been proposed \cite{alan2023parameterized_robust_1,jankovic2018robust_2,nguyen2021robust_3,garg2021robust_4,alan2022disturbance_robust_6,taylor2021towards_robust_7,jankovic2018robust_robust_9,emam2021data_robust_10} that use what is known or assumed about the dynamics not included in the model. To manage these uncertainties, some methods have set maximum uncertainty limits as shown in \cite{jankovic2018robust_cbf_1} and \cite{nguyen2021robust_cbf_2}, but these can be overly cautious. On the other hand, input-to-state safety (ISSf) describes how disturbances can change the safety range. This method tries to lessen the overly cautious approach by setting limits on how much safety can decrease. Yet, even with ISSf techniques, there might still be notable uncertainties as seen in \cite{alan2021safe_iss_cbf_1} and \cite{alan2023control_iss_cbf_2}. There are fewer restrictive adaptive control methods that deal with specific uncertainties \cite{lopez2020robust_adaptive}, but they do not account for disturbances that change over time.

 The Disturbance Observer (DOB) stands out among various methods that model/characterize disturbances. The main purpose of DOB is to efficiently estimate the external disturbances by leveraging the known nonlinear dynamics and states of model that can be measured. This methodology has found extensive use in areas like robotics, automotive, and power electronics \cite{sariyildiz2019disturbance_dist_1,mohammadi2017nonlinear_dist_2,chen2015disturbance_dist_3}. Unlike other robust control strategies that prepare for the worst-case scenarios, methods based on DOB look to minimize the impact of disturbances. They do this by offsetting these disturbances, striking an optimal balance between resilience and efficiency. Most current DOB-focused control strategies target systems where the disturbance's relative degree is at least as significant as the input's relative degree, as seen in \cite{yang2013static_dist_4}. Yet, there are many systems, like missile systems \cite{ginoya2013sliding_dist_5}, flexible joint manipulators [15], and PWM-based DC–DC buck power converters \cite{wang2015extended_dist_6}, where the disturbance's relative degree is lesser. 

The main contributions of the paper are as follows. First, we propose Disturbance Observer based Integral Control Barrier functions (DO-ICBFs) for safe control synthesis of non-affine nonlinear controlled systems with additive bounded disturbances. 
Towards this goal, we first derive an upper bound for the error between the estimated and the actual disturbances. Next, by utilizing these DO-ICBF's and the upper bound for the estimation error, the governing integral control law is modified via the addition of an auxiliary control input that is designed by solving a QP program. Second, we propose High Order DO-ICBFs to extend our approach to nonlinear systems with relative degree that is greater than one with respect to the auxiliary control input and propose High Order DO-ICBFs.

The overall structure of the paper is as follows. Section \ref{sec:prelimaries} discusses the preliminaries followed by the problem statement in Section \ref{sec:problem_statement}. In Section \ref{sec:main_result} we discuss the main results followed by numerical simulations in Section \ref{sec:results}. Finally, we make some concluding remarks in Section \ref{sec:conclusion}.
\section{Nomenclature\label{sec:nomenclature}}
Vectors are denoted by bold symbols. The interior and the boundary of a set $\mathcal{S}$ are denoted by $\text{int}(\mathcal{S})$ and $\partial\mathcal{S}$, respectively. For integers $a$ and $b\geq a$, we denote by $[a;b]_d$, the set of integers $\{a,a+1,\dots,b\}$. By default, for a vector $\boldsymbol{v}$, $\|\boldsymbol{v}\|$ denotes the Euclidean norm. $I_n$ and $0_n$ denote the identity and the zero matrix of dimension $n$, respectively.
\section{Preliminaries\label{sec:prelimaries}}
In this section, we summarize the notions of Control Barrier Functions (CBFs) and Integral Control Barrier Functions (I-CBFs) for dynamically defined control laws. Before we proceed, we first state the invariance lemma.
\begin{lemma}
\normalfont \cite{glotfelter2017nonsmooth_rd_2} Let $b:\left[t_0, t_1\right] \rightarrow \mathbb{R}$ be a continuously differentiable function. If $\dot{b}(t) \geq \alpha(b(t))$, for all $t \in\left[t_0, t_1\right]$, where $\alpha$ is a class $\mathcal{K}$ function, and $b\left(t_0\right) \geq 0$, then $b(t) \geq 0$ for all $ t \in\left[t_0, t_1\right]$.
\label{lemma:invariant_time}
\end{lemma}

Consider a general nonlinear system given by
\begin{align}
    \dot{\bx}=F(\bx,\bu),\quad \bx(0)=\bx_0
    \label{eqn:nonlinear_system_no_dist}
\end{align}
where $\bx\in\mathbb{R}^n$ is the state, $\bu\in\mathbb{R}^m$ is the control input, $F:\mathbb{R}^n\times\mathbb{R}^m\rightarrow\mathbb{R}^n$ is a continuously differentiable function. For a given locally Lipschitz control law $\bu=k(\bx)$ where $k:\mathbb{R}^n\rightarrow\mathbb{R}^m$, let $\Phi^k_F(\bx(0))$ denote the solution to the closed-loop system:
\begin{align}
  \dot{\bx}=F(\bx,k(\bx)), \quad \bx(0)=\bx_0
  \label{eqn:closed_loop_nonlinear_system}
\end{align}
\subsection{Control Barrier Functions (CBFs)\label{subsec:cbf}}
Let $\mathcal{S}\subset\mathbb{R}^n$ denote the safe set and $b(\bx):\mathbb{R}^n\rightarrow\mathbb{R}$ be a continuously differentiable scalar function such that the following holds:
\begin{subequations}
    \begin{align}
    &b(\bx) > 0,\quad\forall \bx\in\text{int}(\mathcal{S})\\
    &b(\bx)=0,\quad\forall \bx\in\partial\mathcal{S}\\
        &b(\bx)<0,\quad\forall \bx\in\mathbb{R}^n\setminus\mathcal{S}
\end{align}
\label{eqn:cbf_defn_nominal}
\end{subequations}
The Control Barrier Function (CBF) is then defined as follows:
\begin{definition}
    \normalfont Let $\mathcal{S}\subset \mathbb{R}^n$ and $b(\bx)$ be defined as in \eqref{eqn:cbf_defn_nominal}. Then, $h$ is
a control barrier function (CBF) for the nonlinear system $\dot{\bx}=F(\bx,\bu)=f(\bx)+g(\bx)\bu$, (where $f:\mathbb{R}^n\rightarrow\mathbb{R}^n$ and $g:\mathbb{R}^n\rightarrow\mathbb{R}^n\times\mathbb{R}^m$ are continuously differentiable functions) if there exists an extended
class-$\mathcal{K}_\infty$\footnote{A continuous function $\gamma$ is said to be class-$\mathcal{K}_\infty$ function if it is continuously increasing, $\gamma(0)=0$ and $\underset{x\rightarrow\infty}{\lim}\;\gamma(x)=0$.} function $\alpha$ such that
\begin{align}
    \underset{\bu\in\mathbb{R}^m}{\sup}\;L_fb(\bx)+L_gb(\bx)\bu\geq -\alpha(b(\bx))
    \label{eqn:sup_cbf_condition}
\end{align}
for all $\bx\in\mathbb{R}^n$.
\end{definition}
We define the set $K_{\text{CBF}}(\bx)$ satisfying \eqref{eqn:sup_cbf_condition} as follows:
\begin{align}
    K_{\text{CBF}}(\bx)=\{\bu\in\mathbb{R}^m:\;\dot{h}(\bx)\geq -\alpha(b(\bx))\}
\end{align}
The following theorem provides the invariance guarantees for controlled systems based on CBF
\begin{theorem}
\normalfont \cite{ames2019control} Consider a subset $\mathcal{S}\subset\mathbb{R}^n$, characterized as the 0-superlevel set of a continuously differentiable function $b:\mathbb{R}^n\rightarrow
\mathbb{R}$. If $b$ is a CBF, and $\frac{\partial b}{\partial\bx}\neq 0$ for every $\bx\in\partial S$, then any controller $\bu=k(\bx)$ with $k(\bx)\in K_{\text{CBF}}(\bx)$ for all $\bx\in\mathcal{S}$ that is Lipschitz continuous ensures that the set $\mathcal{S}$ will be forward invariant under the closed loop dynamics \eqref{eqn:closed_loop_nonlinear_system}. 
\end{theorem}
The synthesis of safe control inputs that guarantee forward invariance for $\mathcal{S}$ can by obtained by solving the following QP:
\begin{subequations}
\begin{align}
  \textbf{CBF-QP}\quad &\bu_S(\boldsymbol{x}):=\underset{\bu\in\mathbb{R}^m}{\argmin}\quad\|\bu-k(\boldsymbol{x})\|^2\\
   &\text{s.t.} \quad L_fb(\bx)+L_gb(\bx)\bu\geq-\alpha(b(\bx)).
   \label{eqn:cbf_constraint}
\end{align}
 \label{eqn:quadratic_problem}
\end{subequations}
where $k(\bx)$ is the nominal feedback stabilizing controller for the system $\dot{\bx}=F(\bx,\bu)=f(\bx)+g(\bx)\bu$ in the absence of safety constraints.
\subsection{Integral Control Barrier Functions (ICBF's)}
Consider the following general feedback law governed by the following ordinary differential equation:
\begin{align}
\dot{\bu}=\phi(\bx, \bu),\quad\quad \bu(0):=\bu_0
\label{eqn:u_governing_icbf}
\end{align}
where $\phi: \mathbb{R}^n \times \mathbb{R}^m \rightarrow \mathbb{R}^m$ is continuously differentiable. The combined system obtained by augmenting the governing state dynamics \eqref{eqn:nonlinear_system} and the input dynamics \eqref{eqn:u_governing_icbf} is given by
\begin{align}
\left[\begin{array}{l}
\dot{\bx} \\
\dot{\bu}
\end{array}\right]=\left[\begin{array}{l}
F(\bx, \bu) \\
\phi(\bx, \bu)
\end{array}\right] ,\quad \left[\begin{array}{l}
\bx(0) \\
\bu(0)
\end{array}\right]=\left[\begin{array}{l}
\bx_0 \\
\bu_0
\end{array}\right]
\label{eqn:integrated_augmeted_system}
\end{align}
Denote by $\bz:=\left[\bx^{\mathrm{T}}, \bu^{\mathrm{T}}\right]^{\mathrm{T}} \in \mathbb{R}^{n+m}$ as the augmented state, which can essentially be considered as the state of the integrated system \eqref{eqn:integrated_augmeted_system}.

Consider $\mathcal{S} \subset \mathbb{R}^n \times \mathbb{R}^m$, termed the safety set that encodes both the state constraints as well as the input constraints. This set is characterized as the 0-superlevel set of a smoothly differentiable function $h:\mathbb{R}^n \times \mathbb{R}^m \rightarrow \mathbb{R}$. Moreover, assume the existence of a class-$\mathcal{K}$ function $\gamma$ such that, for every trajectory $\bz$ of the integrated system \eqref{eqn:integrated_augmeted_system}, the following holds true:
\begin{align}
\dot{h}(\bz)+\gamma(h(\bz)) \geq 0 
\label{eqn:governing_h_dot_augmented_system}
\end{align}
where $\dot{h}(\bz)$ is the time derivative along the system trajectories \eqref{eqn:integrated_augmeted_system}.
Define the vector function $p(\bx, \bu)$ and the scalar function $q(\bx, \bu) $ as follows: 
\begin{align}
&p(\bx, \bu)  :=\left(\frac{\partial h(\bx,\bu)}{\partial \bu}\right)^{\mathrm{T}}\nonumber \\
&q(\bx, \bu):  =-\left(\frac{\partial h(\bx,\bu)}{\partial \bx} F(\bx, \bu)+\frac{\partial h(\bx,\bu)}{\partial \bu} \phi(\bx, \bu)\right.\nonumber \\
& +\gamma(h(\bx, \bu))),\nonumber
\end{align}
The inequality in \eqref{eqn:governing_h_dot_augmented_system} can be equivalently reframed as $q(\bx, \bu) \leq 0$.
Standard CBF techniques are not readily applicable to systems in the form of \eqref{eqn:nonlinear_system} due to the integral control law governed by \eqref{eqn:u_governing_icbf}. The control strategy presented in \cite{ames2020integral_cbf} introduces an auxiliary input into the state space model. This ensures that if the inequality in \eqref{eqn:governing_h_dot_augmented_system} is not satisfied, we can ascertain the minimal alteration of the dynamic control rule to ensure safety. Thus, we adjust the augmented state space model \eqref{eqn:integrated_augmeted_system} to incorporate the input $\bv \in \mathbb{R}^m$ as:
\begin{align}
\left[\begin{array}{l}
\dot{\bx} \\
\dot{\bu}
\end{array}\right]=\left[\begin{array}{l}
F(\bx, \bu) \\
\phi(\bx, \bu)+\bv
\end{array}\right] .
\label{eqn:augmented_system_alteration}
\end{align}
\begin{figure}
    \centering
\includegraphics[width=\columnwidth]{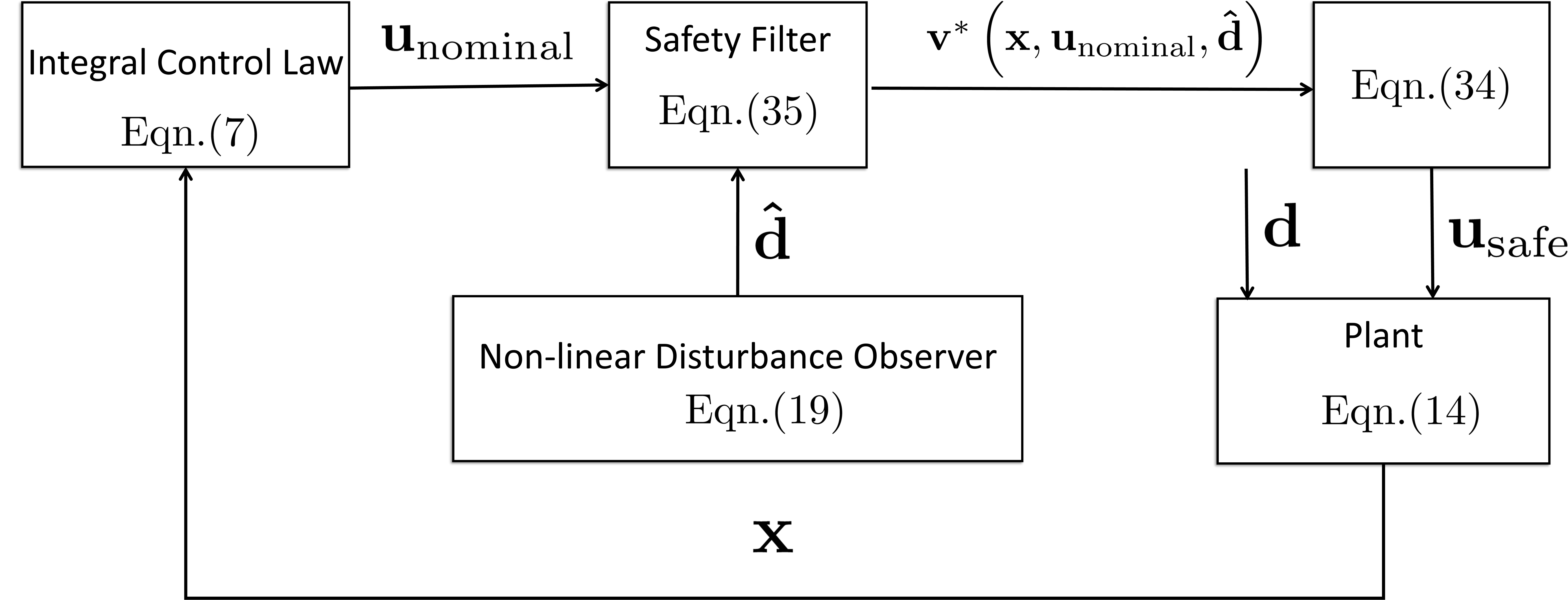}
    \caption{Framework for safe control synthesis of general non-control affine systems with additive bounded disturbances using High Order DO-ICBFs.}
    \label{fig:block_diagram}
\end{figure}
For each $\bz:=\left[\bx^{\mathrm{T}}, \bu^{\mathrm{T}}\right]^{\mathrm{T}} \in \mathbb{R}^n \times \mathbb{R}^m$ and $t \geq 0$, let us define the following set
\begin{align}
K_{\bz}:= & \left\{\bv \in \mathbb{R}^m: \frac{\partial h(\bx,\bu)}{\partial \bx} F(\bx, \bu)\right.\nonumber \\
& \left.+\frac{\partial h(\bx,\bu)}{\partial \bu}(\phi(\bx, \bu)+\bv)+\gamma(h(\bx, \bu)) \geq 0\right\} \\
= & \left\{\bv \in \mathbb{R}^m: p(\bx, \bu)^{\mathrm{T}} \bv \geq q(\bx, \bu)\right\} .\nonumber
\end{align}
Thus, when $ \bv(\bx) $ belongs to $ K_{\bz} $, we can infer that inequality \eqref{eqn:governing_h_dot_augmented_system} holds true. Furthermore, given that the dynamic control law $ \dot{\bu}=\phi(\bx, \bu) $ is intrinsically safe, the conditions $ q(\bx, \bu) \leq 0 $ and $ \bv=0 $ indicate system safety. 

\begin{definition}
\normalfont
 For system \eqref{eqn:nonlinear_system_no_dist} and the associated safe set $ \mathcal{S} \subset \mathbb{R}^n \times \mathbb{R}^m $, defined by the 0-superlevel set of a function $ h: \mathbb{R}^n \times \mathbb{R}^m \rightarrow \mathbb{R} $, i.e. we have $ \mathcal{S}=\{(\bx, \bu) \in \mathbb{R}^n \times \mathbb{R}^m: h(\bx, \bu) \geq 0\} $. The function $ h $ is termed an integral control barrier function (I-CBF) if, for all $ (\bx, \bu) $ in $ \mathbb{R}^n \times \mathbb{R}^m $ and $ t \geq 0 $:
\begin{align}
p(\bx, \bu)=0 \implies q(\bx, \bu) \leq 0 .
\end{align}
\end{definition}
\begin{theorem}
    \normalfont \cite{ames2020integral_cbf} 
Let us assume that a dynamic controller $ \dot{\bu}=\phi(\bx, \bu) $ for the control system \eqref{eqn:nonlinear_system_no_dist} exists. If the safe set $ \mathcal{S} \subset \mathbb{R}^n \times \mathbb{R}^m $ is characterized by an integral control barrier function, $ h: \mathbb{R}^n \times \mathbb{R}^m \rightarrow \mathbb{R} $, then altering the dynamic controller to:
\begin{align}
    \dot{\bu}=\phi(\bx, \bu)+\bv^\star(\bx, \bu)
    \label{eqn:dynamic_controller}
\end{align}
where $ \bv^\star $ is the solution to the QP:
$$
\begin{aligned}
\bv^\star(\bx, \bu) = \underset{\bv \in \mathbb{R}^m}{\operatorname{argmin}}\|\bv\|^2 \\
\text {subject to } p(\bx, \bu)^{\mathrm{T}} \bv \geq q(\bx, \bu)
\end{aligned}
$$
ensures safety, that is, the system \eqref{eqn:nonlinear_system_no_dist} combined with the dynamic controller \eqref{eqn:dynamic_controller} maintains $ \mathcal{S} $ as forward invariant. If $ (\bx_0, \bu_0) $ belongs to $ \mathcal{S} $, then $ (\bx, \bu) $ remains in $ \mathcal{S} $ for all $ t \geq 0 $.
\end{theorem}
Throughout the rest of the paper, we consider the general nonlinear system with bounded additive disturbances as follows:
\begin{align}
\dot{\bx}=F(\bx,\bu)+\ell(\bx)\bd,\quad \bx(0)=\bx_0
    \label{eqn:nonlinear_system}
\end{align}
where $\bx\in\mathbb{R}^n$ is the state, $\bu\in\mathbb{R}^m$ is the control input, $d\in\mathcal{D}\subset\mathbb{R}^p$, $\ell:\mathbb{R}^n\rightarrow\mathbb{R}^n\times\mathbb{R}^p$ and $F:\mathbb{R}^n\times\mathbb{R}^m\rightarrow\mathbb{R}^n$ are continuously differentiable functions. For the sake of brevity, we drop the time indexing of variables $\bx$, $\bu$ and $\bd$ and use them whenever necessary.
\begin{assumption}
    \normalfont The norms of both the disturbance $\bd$ and its rate of change $\dbd$ are confined by known positive constants, that is, for all $t\geq 0$, $\|\bd\| \leq k_0$ and $\|\dbd\| \leq k_1$ where $k_0>0$ and $k_1>0$. 
\label{assumption:bounded_disturbances}
\end{assumption}
\section{Problem Statement\label{sec:problem_statement}}
In this section, we present the problem we address in this paper.
\begin{problem}
    \normalfont Under Assumption \ref{assumption:bounded_disturbances} and given the safe set $\mathcal{S}$ and the nonlinear system \eqref{eqn:nonlinear_system}, for a given estimate of the disturbance $\hbd$, design an integral feedback control law $\bu(\bx,\hbd)$ governed by \eqref{eqn:u_governing_icbf} so that the closed-loop system \eqref{eqn:nonlinear_system} under $\bu(\bx,\hbd)$ is guaranteed to be safe.
\end{problem}
\begin{remark}
\normalfont Guaranteeing safety in the presence of unknown additive unknown bounded disturbances for non-control affine systems still remains an open challenge. One straightforward method to address this case is to consider the worst-case disturbance and apply ICBF. However, this approach can be conservative. Towards this aim, we adopt a disturbance observer based approach to address this problem. This will be illustrated in detail in the following sections. 
\end{remark}
\section{Main Results\label{sec:main_result}}
Consider the output equation given by $\by=\boldsymbol{n}(\bx)$, with $\boldsymbol{n}: \mathbb{R}^n \to \mathbb{R}^m$ being a smooth function. The function $\by$ depends on the initial condition $\bx_0$ and prior control inputs represented as $\bu(\tau)$ for $ \tau \in[0, t)$. The integral control law in \cite{ames2020integral_cbf} is based on the tracking control of Newton-Raphson law which is described as follows.
Now, for any moment $t \geq 0$ with $T > 0$, if we maintain $\bu_t = \bu(t)$ constant over the interval $[t, t+T]$ and carry out a forward integration of \eqref{eqn:nonlinear_system} over this period starting with $\hat{\bx}(t) = \bx(t)$, it provides a forecast of the state $\hat{\bx}(t+T)$. This in turn results in an output prediction:
\begin{align}
    \hat{\by}(t+T) = \zeta(\hat{\bx}(t+T)) =: g(\bx(t), \bu(t))
\end{align}
 Consequently the control input $\bu(t)$ is given by
\begin{align}
\dot{\bu}(t) = \alpha\left(\frac{\partial g}{\partial \bu}(\bx, \bu)^{-1}(\by_{\text{ref}} - \hat{\by}(t+T))\right)
\label{eqn:tracking_control_law}
\end{align}
where $\by_\text{ref}$ is the reference signal. 
However, note that for the $\bu$ to be implementable in real world systems, the inverse of $\frac{\partial g}{\partial \bu}$ must always exist which might not always be the case. To address this limitation, once can consider the following control law based on the Proportional Integral law given as follows:
\begin{align}
    \bu(\by)=K_p(\by-\by_{\text{ref}})+K_I\int_0^t(\by-\by_{\text{ref}})\mathrm{dt}
    \label{eqn:pi_control_law}
\end{align}
Consequently, the integral control law is given by
\begin{align}
        \phi(\bx,\bu)&=\left[K_p\left(\frac{\partial \boldsymbol{n}}{\partial \bx}F(\bx,\bu)-\dot{\by}_{\text{ref}}\right)+K_I(\by-\by_{\text{ref}})\right]
        \label{eqn:pi_dot_control_law}
\end{align}
In contrast to the control law \eqref{eqn:tracking_control_law}, the PI based control law \eqref{eqn:pi_control_law} would be devoid of any singularity. However, note that the PI based controller would not be always the most appropriate controller especially if the system is highly nonlinear. This nonlinearity can cause challenges for a PI controller, which is essentially a linear control strategy. While PI controllers can handle certain nonlinearities, general non-affine systems often can have complex behaviors that a simple PI controller cannot address. Other advanced control methods, such as backstepping or Nonlinear Model Predictive Control (MPC), can be leveraged specifically to address challenges in non-affine systems.
\subsection{Disturbance Observer\label{subsec:nonlinear_disturbance_observer}}
We establish the nonlinear Disturbance Observer (DOB) for the general nonlinear system \eqref{eqn:nonlinear_system} that will play a key role in the safe control design that will be presented in Section \ref{subsec:do-icbf}. For system \eqref{eqn:nonlinear_system}, we propose the DOB as follows:
\begin{align}
\left\{\begin{array}{l}
\hbd=\br+\beta \boldsymbol{q}, \\
\dot{\br}=-\beta L_d\left(F(\bx,\bu)+\ell(\bx) \hbd\right),
\end{array}\right.
\label{eqn:dob_governing_dist}
\end{align}
where $\beta>0$ is a positive tuning parameter. Here, $\hbd$ represents the estimated disturbance, while $L_d(\bx)$ is the observer gain ensuring $-\bx^{\mathrm{T}} L_d \ell(\bx) \bx \leq-\bx^{\mathrm{T}} \bx$ for any $\bx$ (for instance, $L_d=$ $-\left(\ell(\bx)^{\mathrm{T}} \ell(\bx)\right)^{-1} \ell(\bx)^{\mathrm{T}}$ if $\ell(\bx)$ maintains a full column rank). $q(\bx)$ is a function satisfying $\frac{\partial \boldsymbol{q}}{\partial \bx}=L_d(\bx)$. The computation of $q(\bx)$ and $L_d(\bx)$ is in general non-trivial and depends on the specific problem \cite{sariyildiz2019disturbance_dist_1,mohammadi2017nonlinear_dist_2,chen2015disturbance_dist_3}. 
\subsection{Disturbance Observer based Integral Control Barrier Functions (DO-ICBFs)\label{subsec:do-icbf}}
Incorporating the disturbance observer in \eqref{eqn:integrated_augmeted_system}, we now consider the augmented system as follows:
\begin{align}
    \left[\begin{array}{l}
\dot{\bx} \\
\dot{\bu}\\
\dot{\br}
\end{array}\right]=\left[\begin{array}{l}
F(\bx, \bu)+\ell(\bx)\bd \\
\phi(\bx, \bu)+\bv\\
-\beta L_d\left(F(\bx,\bu)+\ell(\bx) \hbd\right)
\end{array}\right] 
\label{eqn:augmented_system_with_d_disturbance}
\end{align}
where $\hbd=\br+\beta \boldsymbol{q}$. Denote the augmented state by $\bz^d:=[\bx^\mathrm{T},\;\bu^\mathrm{T},\;\br^\mathrm{T}]^\mathrm{T}$. Further define $w(\bx,\bu,\hbd)$ as follows:
\begin{align}
   & w(\bx, \bu,\hbd): \nonumber\\
    &=-\left(\frac{\partial h(\bx,\bu)}{\partial \bx} F(\bx, \bu)+\frac{\partial h(\bx,\bu)}{\partial \bx} \ell(\bx)\hbd\right.+\nonumber
    \\&\quad\quad\frac{\partial h(\bx,\bu)}{\partial \bu} \phi(\bx, \bu) +\gamma(h(\bx, \bu))),
\end{align}

\begin{definition}
\normalfont (DO-ICBFs)
 Under Assumption 1, for system \eqref{eqn:nonlinear_system} and the associated safe set $ \mathcal{S} \subset \mathbb{R}^n \times \mathbb{R}^m $, defined by the 0-superlevel set of a function $ h: \mathbb{R}^n \times \mathbb{R}^m \rightarrow \mathbb{R} $. The function $ h $ is termed a Disturbance Observer based Integral Control Barrier function (DO-ICBF) if, for all $ (\bx, \bu) $ in $ \mathbb{R}^n \times \mathbb{R}^m $ and $ t \geq 0 $:
\begin{align}
p(\bx, \bu)=0 \implies w(\bx, \bu,\hbd) \leq -c(\bx,\bu,h,t) ,
\label{eqn:doicbf_main_condition}
\end{align}
where $c(\bx,\bu,h,t)>0$ for all $t\geq 0$ and which is given by
\begin{align}
    &c(\bx,\bu,h,t)=\nonumber\\
    &\underbrace{\left\|\frac{\partial h(\bx,\bu)}{\partial \bx}\ell(\bx)\right\|}_{c_1(\bx,\bu,h)}\underbrace{\sqrt{\frac{2 \mu_1 \lambda\left\|\ed(0)\right\|^2 e^{-2 \lambda t}+k_1^2\left(1-e^{-2 \lambda t}\right)}{2 \mu_1 \lambda}} }_{c_2(t)}
\end{align}
for some positive constants $\lambda\triangleq \beta-\frac{\mu_1}{2}$, $0<\mu_1<2\beta$ and $k_1>0$.
\end{definition}
\begin{theorem}
    \normalfont 
Under Assumption \ref{assumption:bounded_disturbances}, for the control system \eqref{eqn:nonlinear_system}, assume a dynamic controller exists as $ \dot{\bu}=\phi(\bx, \bu) $. If the safe set $ \mathcal{S} \subset \mathbb{R}^n \times \mathbb{R}^m $ is characterized by a DO-ICBF, $ h: \mathbb{R}^n \times \mathbb{R}^m \rightarrow \mathbb{R} $, then modifying the dynamic controller to:
$$
\dot{\bu}=\phi(\bx, \bu)+\bv^\star(\bx, \bu,\hbd)
$$
where $ \bv^\star $ is the solution to the QP:
\begin{subequations}
\begin{align}
&\bv^\star(\bx, \bu,\hbd) = \underset{\bv \in \mathbb{R}^m}{\operatorname{argmin}}\|\bv\|^2 \\
\text {subject to } &p(\bx, \bu)^{\mathrm{T}} \bv \geq w(\bx, \bu,\hbd)+c(\bx, \bu,t)
\end{align}
\label{eqn:qp_for_do_icbf}
\end{subequations}
ensures safety. In other words, the system \eqref{eqn:augmented_system_with_d_disturbance} combined with $\bv^\star$ obtained from \eqref{eqn:qp_for_do_icbf} maintains $ \mathcal{S} $ as forward invariant, that is, if $ (\bx_0, \bu_0) $ belongs to $ \mathcal{S} $, then $ (\bx, \bu) $ remains in $ \mathcal{S} $ for all $ t \geq 0 $.
\end{theorem}
\begin{proof}
If condition \eqref{eqn:doicbf_main_condition} is satisfied, the QP is feasible. Furthermore, the DO-ICBF condition $\dot{h}(\bz^d)+\gamma(h(\bz^d))\geq 0$ turns out to be as follows i.e. $h(\bx,\bu)$ along system trajectories \eqref{eqn:augmented_system_with_d_disturbance} (with $\bv=0$) is given by
\begin{align}
&\dot{h}(\bz^d)+\gamma(h(\bz^d))\nonumber\\
& =\left(\frac{\partial h(\bx,\bu)}{\partial \bx} F(\bx, \bu)+\frac{\partial h(\bx,\bu)}{\partial \bx} \ell(\bx){d}\right.+\nonumber
    \\&\quad\quad\frac{\partial h(\bx,\bu)}{\partial \bu} \phi(\bx, \bu) +\gamma(h(\bx, \bu)))\geq 0\nonumber\\
    &=\left(\frac{\partial h(\bx,\bu)}{\partial \bx} F(\bx, \bu)+\frac{\partial h(\bx,\bu)}{\partial \bx} \ell(\bx)\hbd\right.+\nonumber
    \\&\quad\quad\frac{\partial h(\bx,\bu)}{\partial \bu} \phi(\bx, \bu) +\gamma(h(\bx, \bu))) \nonumber\\
    &\geq\quad\frac{\partial h(\bx,\bu)}{\partial \bx} \ell(\bx)(\hbd-\bd)\geq \left\|\frac{\partial h(\bx,\bu)}{\partial \bx} \ell(\bx)(\hbd-\bd)\right\|\nonumber
\end{align}
Denote the error in disturbance estimation $\ed$ as
\begin{align}
    \ed=\hbd-\bd
\label{eqn:error_definition}
\end{align}
The derivative of this error $\dot{\boldsymbol{e}}_d$ can be written as
\begin{align}
    \dot{e}_d&=\dot{\hat{\boldsymbol{d}}}-\dbd=\dot{\br}+\beta \frac{\partial q}{\partial \bx} \dot{\bx}-\dbd=\dot{\br}+\beta L_d \dot{\bx}-\dbd\nonumber
\end{align}
 Now, consider a candidate Lyapunov function $V_1=\frac{1}{2}\left\|\ed\right\|^2$. Using the definition of $L_d$, the expression for $\dot{V}_1$ is given by
\begin{align}
\dot{V}_1&=\boldsymbol{e}^\mathrm{T}_d\dot{\boldsymbol{e}}_d\nonumber\\
&=(\hbd-\bd)^\mathrm{T}(\dot{\br}+\beta L_d(F(\bx,\bu)+\ell(\bx)\bd)-\dbd)\nonumber
\end{align}
 By invoking Assumption \ref{assumption:bounded_disturbances} using the expression of $\dot{\br}$ , we have
\begin{align}
 \dot{V}_1 &=(\hbd-\bd)^\mathrm{T}(-\beta L_d\ell(\bx))(\hbd-\bd)-(\hbd-\bd)^\mathrm{T}\dbd\nonumber\\
&\leq -\beta \|\ed\|^2-(\hbd-\bd)^\mathrm{T}\dbd\nonumber\\
&=-\left(\beta-\frac{\mu_1}{2}\right)\|\ed\|^2-\mu_1\|\ed\|^2-(\hbd-\bd)^\mathrm{T}\dbd\nonumber  
\end{align}
where $\mu_1$ is a constant within the range $0<\mu_1<2 \beta$. Using the fact that $k_1\left\|\ed\right\| \leq \frac{\mu_1}{2}\left\|\ed\right\|^2+\frac{1}{2 \mu_1} k_1^2$, we have,
\begin{align}
\dot{V}_1 \leq-2 \lambda V_1+\frac{k_1^2}{2 \mu_1}  
\label{eqn:V_dot_upper_bound}
\end{align}
where $\lambda$ is defined as $\lambda \triangleq \beta-\frac{\mu_1}{2}$.  Consequently, an upper bound for $\|\ed\|$ is given by
\begin{align}
    \left\|\ed\right\| \leq \sqrt{\frac{2 \mu_1 \lambda\left\|\ed(0)\right\|^2 e^{-2 \lambda t}+k_1^2\left(1-e^{-2 \lambda t}\right)}{2 \mu_1 \lambda}} 
    \label{eqn:error_dist_uub}
\end{align}
Consequently, the DO-ICBF condition $\dot{h}(\bz^d)+\gamma(h(\bz^d))\geq 0$ becomes
\begin{align}
    w(\bx, \bu,\hbd) \leq -c(\bx,\bu,h,t)\nonumber
\end{align}
Hence the theorem follows.
\end{proof}
It can be easily verified that the optimal solution for \eqref{eqn:qp_for_do_icbf} via the KKT conditions is given by
\begin{align}
v^\star(\bx, \bu, \hbd)= \begin{cases}\frac{f(\bx, \bu,\hbd,t)}{\|p(\bx, \bu)\|^2} p(\bx, \bu) & \text { if } f(\bx, \bu,\hbd,t)  >0 \\ 0 & \text { if }f(\bx, \bu,\hbd,t) \leq 0\end{cases}\nonumber
\end{align}
where $f(\bx, \bu,\hbd,t)=w(\bx, \bu,\hbd) +c(\bx,\bu,h,t)$.
\begin{remark}
    \normalfont Note that $\underset{t\rightarrow \infty}{\lim}\;c(\bx,\bu,h,t)=\frac{c_1(\bx)k_1}{\sqrt{2\mu_1\lambda}}$. This means that if the disturbance observer is given enough time, the estimation error will become uniformly ultimately bounded.
\end{remark}
\begin{remark}
    \normalfont In many real-world applications, the state and input constraints are usually encoded via separate scalar functions. For instance $h_1(\bx,\bu,\hbd)=a_1^2-\|\bx\|^2$ represents constraining the states inside the ball of radius $a_1$ and $h_2(\bx,\bu,\hbd)=a_2^2-\|\bx\|^2$ represents constraining the ball constraint for the control inputs. Furthermore, it might be difficult to come up with a single scalar function $h(\bx,\bu)$ that could encode both the state and input constraints. In that case, $h_1(\bx,\bu)$ might not be a valid DO-ICBF. Towards this aim, \cite{ames2020integral_cbf} proposes to use $h_e(\bx,\bu)=\dot{h}_x(\bx)+\alpha(h_x(\bx))$ as a ICBF. However it might be the case that $\frac{\partial h_e(\bx,\bu)}{\partial \bu}=0$ does not imply that $\frac{\partial h_e(\bx,\bu)}{\partial \bx}F(\bx,\bu)+\alpha(h_e(\bx,\bu))\geq 0$ for some $(\bx,\bu)\in\mathbb{R}^n\times\mathbb{R}^m$.  To address this drawback, we propose a general framework termed High Order DO-ICBFs. This is illustrated by considering the following simple example.
\end{remark}
\begin{example}
    \normalfont Consider the following example where $d=0$
    \begin{align}
        \dot{x}=x-u^2\nonumber
    \end{align}
    The state constraint set (or safe set) $\mathcal{X}=\{x:x\leq 4\}$ and the input constraint $\mathcal{U}=[-1,\;1]$ can be encoded by functions $h_x(x,u)$ and $h_u(x,u)$ as follows:
    \begin{align}
        h_x(x,u)=4-x,\quad\quad h_u(x,u)=1-u^2\nonumber
    \end{align}
    In this case, the DO-ICBF conditions give:
    \begin{align}
        \begin{bmatrix}
            \frac{\partial h_x(x,u}{\partial u}\\
            \frac{\partial h_u(x,u)}{\partial u}
        \end{bmatrix}=\begin{bmatrix}
            0\\
            -2u
        \end{bmatrix}=0\implies \dot{h}_x(x,u=0)=
            -x   .         \nonumber
    \end{align}
    Therefore, any trajectory starting from $x=4$ would go out of the safe set $\mathcal{X}$. Consequently, the $h_x(x,u)$ and $h_u(x,u)$ do not represent a valid ICBF or a valid DO-ICBF.
\end{example}
\subsection{High Order DO-ICBFs\label{subsec:high_order_do_icbfs}}
For the safe set $\mathcal{S}$ defined over the joint space of state and the input, it might be the case that
\begin{align}
p(\bx,\bu)=0\centernot\implies w(\bx, \bu,\hbd) \leq -c(\bx,\bu,h,t) 
\end{align}

\begin{figure*}[h!]
\centering
\begin{subfigure}[t]{0.32\linewidth}
    \centering
    \includegraphics[width=1\linewidth]{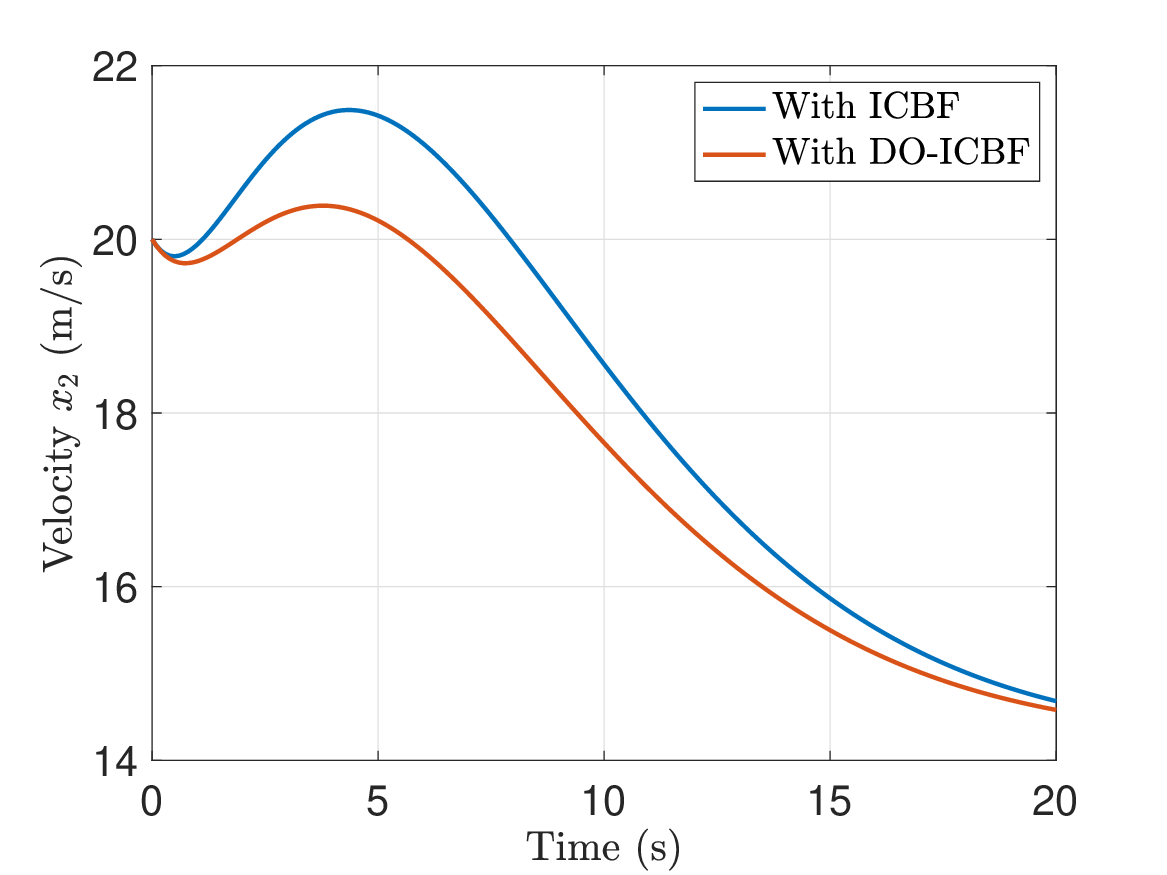}
\caption{Velocity vs time}
    \label{fig:velocity_acc}
\end{subfigure}
\begin{subfigure}[t]{0.32\linewidth}
    \centering
    \includegraphics[width=1\linewidth]{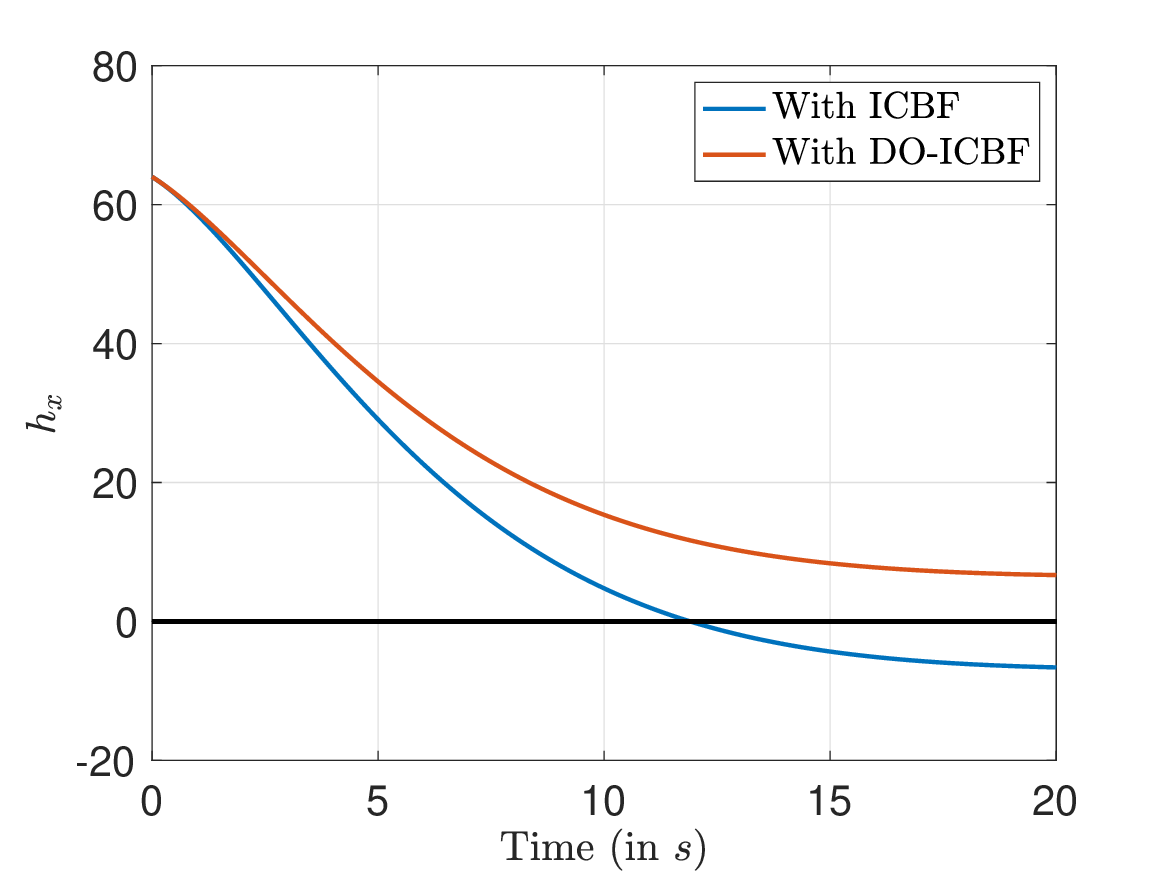}
 \caption{$h_x$ versus time}
    \label{fig:safety_acc}       
\end{subfigure}
\begin{subfigure}[t]{0.32\linewidth}
    \centering
    \includegraphics[width=1\linewidth]{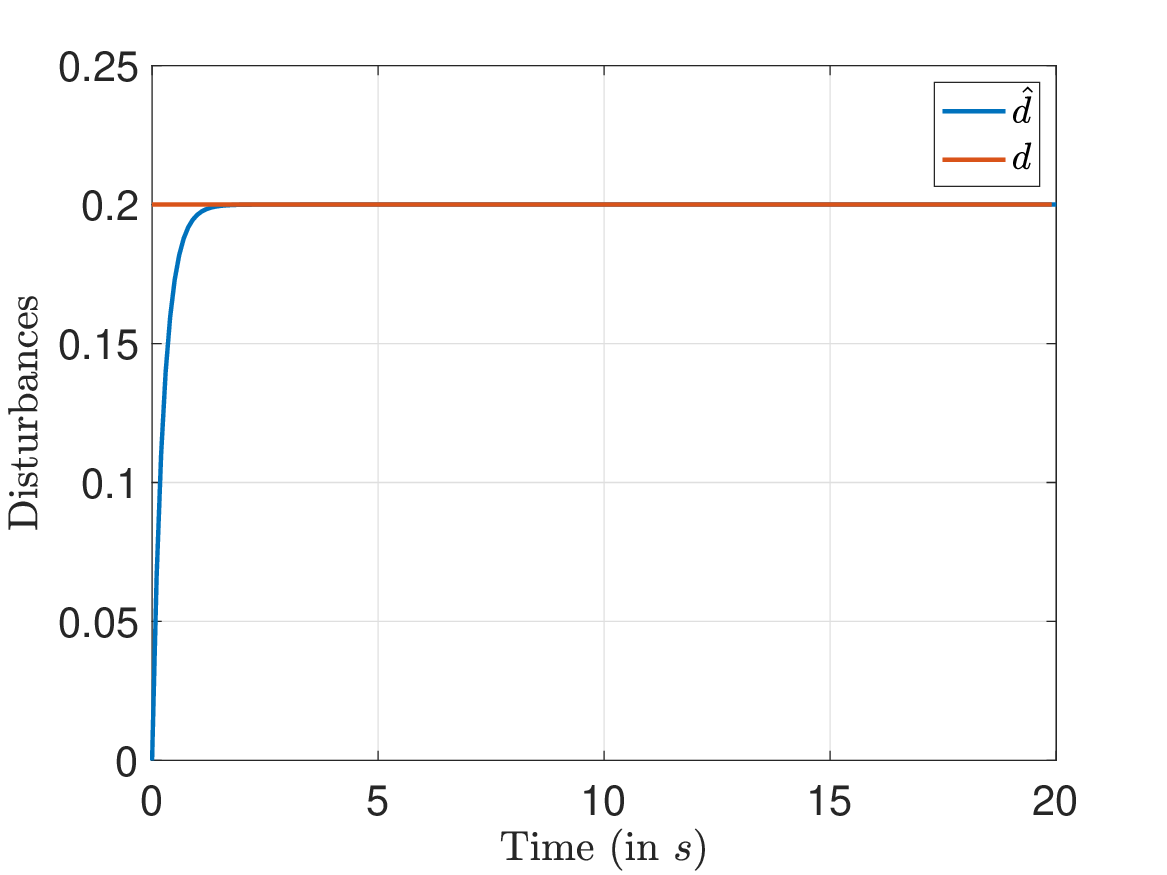}
\caption{Evolution of estimated and real disturbance with time}
    \label{fig:dist_acc}      
\end{subfigure}
 \caption{As observed from Fig. \ref{fig:velocity_acc}, the velocity of vehicle converges to the desired velocity $v_0$ using both ICBF as well as DO-ICBF. However, as shown in Fig. \ref{fig:safety_acc}, safety (i.e. $h_x\geq 0$) is guaranteed using the proposed DO-ICBF. Finally, Fig. \ref{fig:dist_acc} shows that the estimated disturbance $\hbd$ converges to the real disturbance $\bd$.}
\label{fig:}
\end{figure*}

In other words, the function $h(\bx,\bu)$ would not be a valid DO-ICBF. To address this limitation, we propose a subset of set $\mathcal{S}$ i.e $\mathcal{S}^m\subset\mathcal{S}$ which would render $\mathcal{S}^m$ forward invariant. Consider the  sequence of functions $b_i(\bx,\bu)$ for all $i\in[0;m]_d$ as follows
\begin{align}
    &b_0(\bx,\bu):=h(\bx,\bu),\nonumber\\ 
&b_1(\bx,\bu,\hbd):=\dot{b}_0(\bx,\bu)+\gamma_1(b_0(\bx,\bu))-c(\bx,\bu,b_0,t),\nonumber\\
&\quad=\underbrace{p(\bx,\bu)v}_{=0}-q(\bx,\bu)+\gamma_1(b_0(\bx,\bu)-c(\bx,\bu,b_0,t))\nonumber\\
&\quad=-q(\bx,\bu)+\gamma_1(b_0(\bx,\bu)-c(\bx,\bu,b_0,t))\nonumber\\
&b_2(\bx,\bu,\hbd):=\dot{b}_1(\bx,\bu)+\gamma_2(b_1(\bx,\bu))-c(\bx,\bu,b_1,t),\nonumber\\
 &\quad\quad   \vdots\nonumber\\
&b_m(\bx,\bu,\hbd):=\dot{b}_{m-1}(\bx,\bu)+\gamma_m(b_{m-1}(\bx,\bu))\nonumber\\
&\quad\quad\quad\quad\quad\quad-c(\bx,\bu,b_{m-1},t)
\end{align}
where $\gamma_1,\;\gamma_2,\dots,\gamma_m$ are class-$\mathcal{K}$ functions. We assume that the functions $F$ and $\ell(\bx)$ are sufficiently smooth such that $b_m$ and its derivative are defined. 
\begin{definition}
    \normalfont The relative degree $m\geq 1$ for system \eqref{eqn:integrated_augmeted_system} is defined as follows:
\begin{align}
    &\frac{\partial b_i(\bx,\bu,\hbd)}{\partial \bu}=0\centernot\implies q_i(\bx, \bu)\leq-c(\bx,\bu,b_{i},t)\;\nonumber\\
&\quad\quad\quad\quad\quad\quad\quad\quad\quad \quad\quad\forall\; i\in[0;m-1]_d\nonumber\\
    &\frac{\partial b_i(\bx,\bu,\hbd)}{\partial \bu}=0\implies q_i(\bx, \bu) \leq -c(\bx,\bu,b_{i},t)\nonumber\\
    &\quad \quad\quad\quad \quad\quad\quad \quad \quad\quad\quad\quad\;\text{for}\;\; i=m
\end{align}
where $q_i(\bx, \bu,\hbd)$ for $i\in[0;m]_d$ is given by
\begin{align}
&q_i(\bx, \bu,\hbd)\nonumber\\
    &=-\left(\frac{\partial b_{i}(\bx,\bu)}{\partial \bx} F(\bx, \bu)+\frac{\partial b_{i}(\bx,\bu)}{\partial \bx} \ell(\bx)\hbd\right.+\nonumber
    \\&\quad\quad\frac{\partial b_{i}(\bx,\bu)}{\partial \bu} \phi(\bx, \bu) +\gamma_i(b_{i}(\bx, \bu))),
\end{align}
In this paper, we assume that such a positive integer $m$ exists
\end{definition}
Now, construct a sequence of sets as follows:
\begin{align}
    \mathcal{S}_i=\{(\bx,\bu)\;:b_i(\bx,\bu,\hbd)\geq 0\},\quad\forall\quad i\in[0;m]_d
    \label{eqn:S_i}
\end{align}
\begin{definition}
\normalfont (High Order DO-ICBFs)
 Under Assumption \ref{assumption:bounded_disturbances}, for system \eqref{eqn:nonlinear_system} and the associated safe set $ \mathcal{S}^m:=\cap_{j=0}^m \mathcal{S}_j $, if the system \eqref{eqn:integrated_augmeted_system} has relative degree $m\geq 1$ and if there exists a function $b_m $ such that 
 \begin{align}
\frac{\partial b_{m}(\bx,\bu) }{\partial \bu}=0 \implies q_m(\bx, \bu) \leq -c(\bx,\bu,b_{m},t) ,
\end{align}
 then the function $b_m$ is termed a High Order DO-ICBF.
\end{definition}

\begin{theorem}
    \normalfont  
 For the controlled system \eqref{eqn:nonlinear_system}, assume a dynamic controller $ \dot{\bu}=\phi(\bx, \bu) $ exists. If the safe set $ \mathcal{S}^m:=\mathcal{S}_0\cap \mathcal{S}_1\dots\cap\mathcal{S}_m \subset \mathbb{R}^n \times \mathbb{R}^m $ where $S_i$ for all $i\in[0;m]$ is defined as in \eqref{eqn:S_i} is characterized by an integral control barrier function, $ h: \mathbb{R}^n \times \mathbb{R}^m \rightarrow \mathbb{R} $, then altering the dynamic controller to:
\begin{align}
\dot{\bu}=\phi(\bx, \bu)+\bv^\star(\bx, \bu,\hbd)
\label{eqn:high_order_do_icbf_u}
\end{align}

where $ \bv^\star $ is the solution to the QP:
\begin{subequations}
\begin{align}
&\bv^\star(\bx, \bu,\hbd) = \underset{\bv \in \mathbb{R}^m}{\operatorname{argmin}}\|\bv\|^2 \\
\text {subject to } &\frac{\partial b_{m}(\bx,\bu) }{\partial \bu}^{\mathrm{T}} \bv \geq w_m(\bx, \bu,\hbd)+c(\bx, \bu,b_{m-1},t)
\end{align}
\label{eqn:qp_for_high_order_do_icbf}
\end{subequations}
where $w_m(\bx,\bu,\hbd)$ is defined as follows:
\begin{align}
   & w_m(\bx, \bu,\hbd): \nonumber\\
    &=-\left(\frac{\partial b_{m}(\bx,\bu)}{\partial \bx} F(\bx, \bu)+\frac{\partial b_{m}(\bx,\bu)}{\partial \bx} \ell(\bx)\hbd\right.+\nonumber
\\&\quad\quad\frac{\partial h(\bx,\bu)}{\partial \bu} \phi(\bx, \bu) +\gamma_m(h(\bx, \bu))),
\end{align}
ensures that the set $\mathcal{S}^m$ is forward invariant. In other words, the controlled system \eqref{eqn:nonlinear_system} combined with the dynamic controller \eqref{eqn:high_order_do_icbf_u} maintains $ \mathcal{S}^m:=\mathcal{S}_0\cap \mathcal{S}_1\dots\cap\mathcal{S}_m \subset \mathbb{R}^n \times \mathbb{R}^m $ as forward invariant. If $ (\bx_0, \bu_0) $ belongs to $ \mathcal{S}^m $, then $ (\bx, \bu) $ remains in $ \mathcal{S}^m $ for all $ t \geq 0 $.
\label{thm:high_order_do_icbf}
\end{theorem}
\begin{proof}
   If $b_m(\bx,\bu,\hbd)$ is a High Order DO-ICBF, then $b_m(\bx,\bu,\hbd) \geq 0$ for all $t\geq 0$ , i.e., 
   \begin{align}
    \dot{b}_{m}(\bx,\bu,\hbd)+\gamma_m\left(b_{m}(\bx,\bu,\hbd)\right) \geq c(\bx,\bu,b_{m},t)   \nonumber
   \end{align}
   By Lemma \ref{lemma:invariant_time}, since $\bx_0 \in \mathcal{S}_{m-1}$ (i.e., $b_{m-1}\left(\bx_0,\bu_0\right) \geq 0$, and $b_{m-1}(\bx,\bu,\hbd)$ is an explicit form of $\left.b_{m-1}(t)\right)$, then $b_{m-1}(\bx,\bu,\hbd) \geq 0$ for all $t\geq 0$, i.e., 
   \begin{align}
     \dot{b}_{m-1}(\bx,\bu,\hbd)+\gamma_{m-1}\left(b_{m-1}(\bx,\bu,\hbd)\right) \geq c(\bx,\bu,b_{m-1},t)  \nonumber
   \end{align}
   Again, by Lemma \ref{lemma:invariant_time}, since $\bx_0 \in \mathcal{S}_{m-2}$, we also have 
   \begin{align}
       b_{m-2}(\bx,\bu,\hbd) \geq 0\,\quad \forall\;t\geq 0.\nonumber
   \end{align}
   Iteratively, we can show that $\bx \in \mathcal{S}_i$ for all $ i \in[0;m]_d$ and $t\geq 0$. Therefore, the sets $\mathcal{S}_0, \mathcal{S}_1, \ldots, \mathcal{S}_m$ are forward invariant. Consequently the set $\mathcal{S}:=\cap_{j=0}^m \mathcal{S}_j$ is forward invariant.
\end{proof}
\begin{remark}
    \normalfont The optimal solution for \eqref{eqn:qp_for_high_order_do_icbf} via the KKT conditions is given by
\begin{align}
\bv^\star(\bx, \bu, \hbd))= \begin{cases}\frac{f_m(\bx, \bu,\hbd,t)}{\left\|\frac{\partial b_{m}(\bx,\bu) }{\partial \bu}\right\|^2} \frac{\partial b_{m}(\bx,\bu) }{\partial \bu} &\nonumber\\
\quad\quad\quad\quad\quad\quad\quad\quad\quad\text { if } f_m(\bx, \bu,\hbd,t)  >0 \\ 0 \quad\quad\quad\quad\quad\quad\quad\quad\;\;\text { if }f_m(\bx, \bu,\hbd,t) \leq 0\end{cases}\nonumber
\end{align}
where $f_m(\bx, \bu,\hbd,t)=w_m(\bx, \bu,\hbd) +c(\bx,\bu,b_{m-1},t)$. It can be easily shown that the function $\bv^\star(\bx, \bu, \hbd)$ is smooth.
\end{remark}
\begin{table}[]
    \centering
\begin{tabular}{|c|c|c|c|}
\hline Parameter & Value & Parameter & Value \\
\hline \hline $\mathrm{g}$ & 9.81 & $v_0$ & 13.89  \\
\hline $m$ & 1650 & $\alpha$ & 10 \\
\hline$c_0$ & 0.1 &$\gamma$ & 1\\
\hline$c_1$ & 5 & $c$ & 0.3  \\
\hline$c_2$ & 0.25 & $v_d$ & 24 \\
\hline$\beta$ & 1 & $L_d(\bx)$ & $I_3$ \\
\hline
\end{tabular}
\caption{Parameters for the adaptive cruise control problem}
\label{table:parameters}
\end{table}
\section{Results\label{sec:results}}
\begin{figure*}[h!]
 \centering
 \begin{subfigure}[t]{0.31\textwidth}
{\includegraphics[width=1\linewidth]{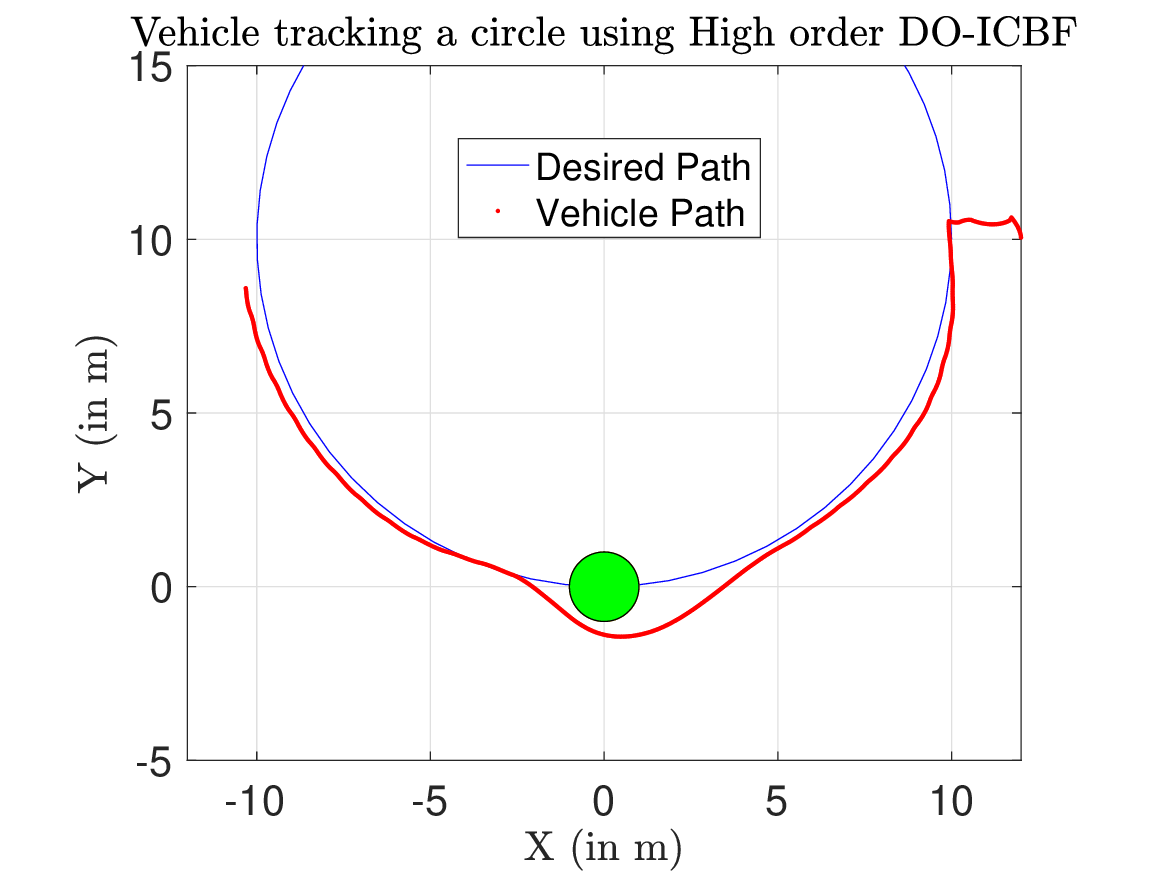}}
\caption{$x$ vs $y$}
\label{fig:stanley_x_and_y}
 \end{subfigure}
 \begin{subfigure}[t]{0.31\textwidth}
{\includegraphics[width=1\linewidth]{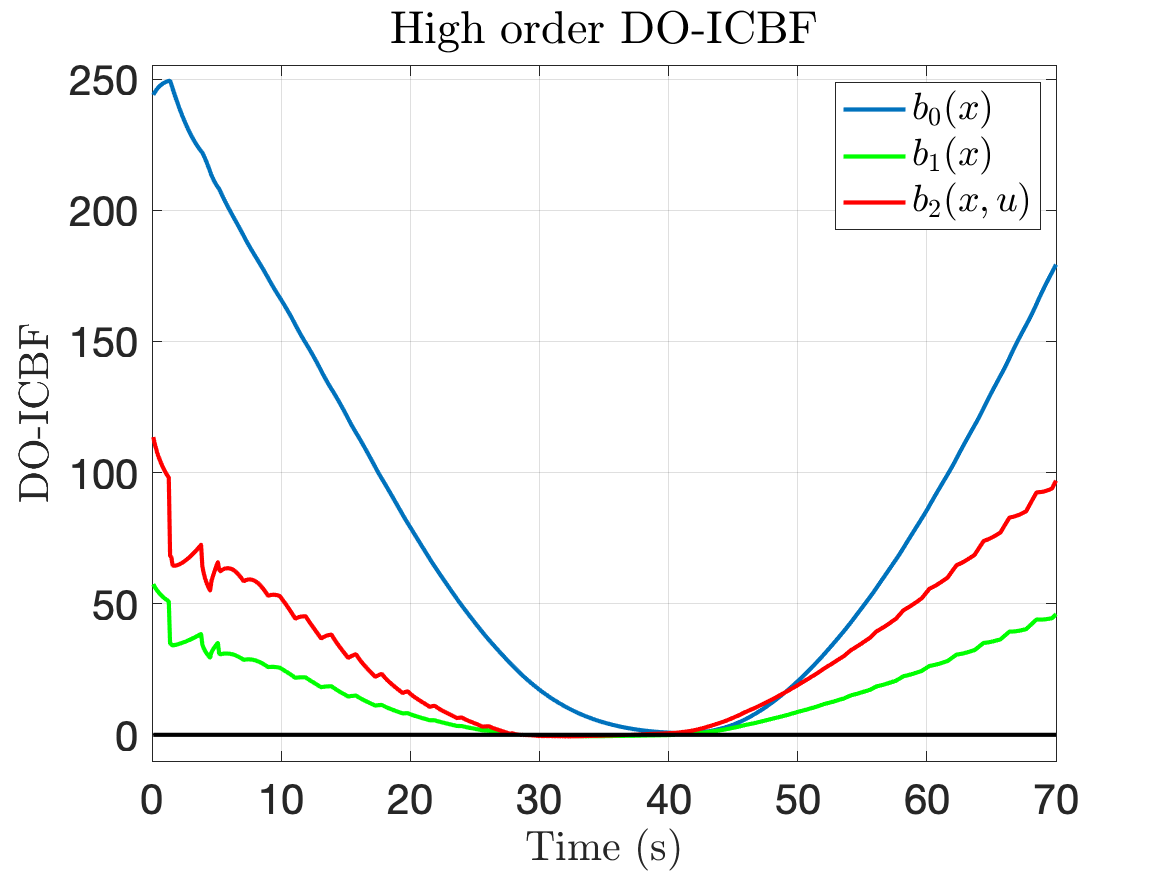}}
 \caption{Evolution of $b_0(\bx)$, $b_1(\bx)$, and $b_2(\bx,\bu)$ versus time}
\label{fig:b_2_variation}
 \end{subfigure}
 \begin{subfigure}[t]{0.31\textwidth}
{\includegraphics[width=1\linewidth]{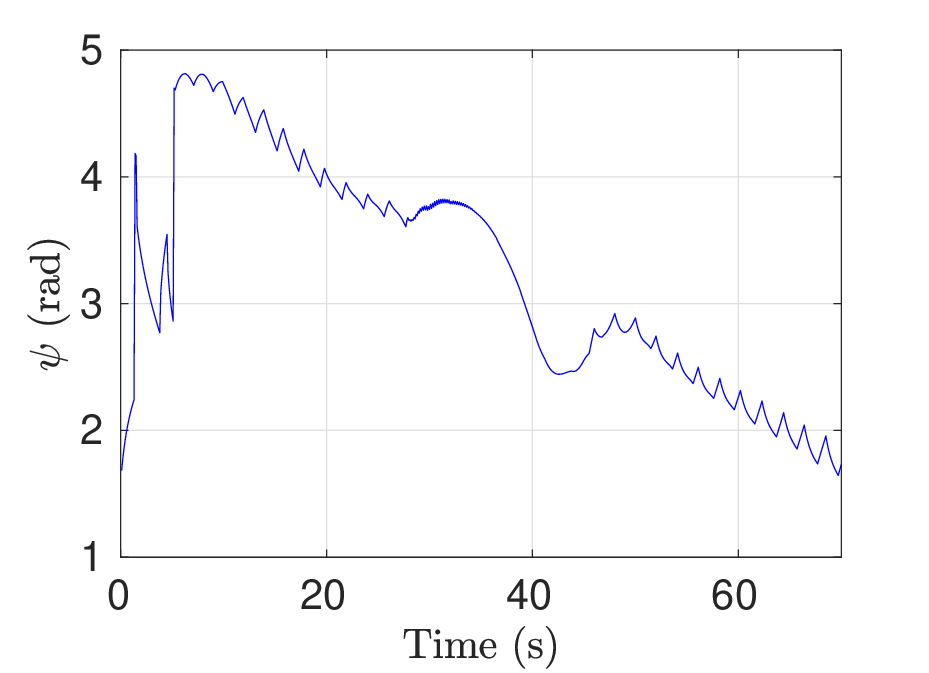}}
\caption{Evolution of yaw angle $\psi$ with time}
\label{fig:psi_variaton}
\end{subfigure}
\caption{As shown in Fig. \ref{fig:stanley_x_and_y}, the vehicle governed by non-control affine dynamics avoids the green obstacle using the second order DO-ICBF. In addition, Fig. \ref{fig:b_2_variation} shows that $b_i(\bx,\bu)$ for $i\in[0;2]_d$ remains positive throughout the travel duration thereby validating Theorem \ref{thm:high_order_do_icbf}. Finally, Fig. \ref{fig:psi_variaton} shows the variation of $\psi$ with time $t$. }
\label{fig:vehicle_tracking_using_stanley_controller}
\end{figure*}
In this section, we consider two examples to validate the proposed DO-ICBF and High Order DO-ICBF's. First, is the standard control affine adaptive cruise control problem with additive bounded disturbances where the vehicle is considered to be a point mass. Second, to illustrate High Order DO-ICBF's we consider a more general non-control affine vehicle dynamics (with no additive bounded disturbances) which is based on the bicycle model and incorporates the Ackermann steering dynamics as well. The source code is available at \href{https://github.com/Vrushabh27/High-Order-DO-ICBF}{https://github.com/Vrushabh27/High-Order-DO-ICBF}
\subsection{Adaptive Cruise Control (ACC)\label{subsec:adaptive_control}}
 The system's dynamics can be expressed as:
\begin{align}
\dot{\bx}=\left[\begin{array}{c}
x_2 \\
-\frac{1}{m} F_r(\bx) \\
v_0-x_2
\end{array}\right]+\left[\begin{array}{c}
0 \\
\frac{1}{m} \\
0
\end{array}\right] u+\left[\begin{array}{c}
0 \\
1 \\
0
\end{array}\right]d
\label{eqn:acc_dynamics}
\end{align}

Here, $\bx= [x_1, x_2]^\mathrm{T} $ denotes the car's position and speed respectively, $ m $ denotes the vehicle's mass, $ x_3 $ is the gap between this vehicle and the leading one moving at a speed $ v_0 $ and $d=2\mathrm{m/s}$ is the constant additive bounded disturbance which is assumed to be unknown. The term $ F_r(\bx) $ denotes the empirical rolling resistance, defined as $ F_r(\bx) = c_0+c_1 x_2+c_2 x_2^2 $. The objective is to steer the vehicle to a target velocity i.e. $x_2 \to v_d $, which is articulated as an output: $ y=n(\bx)=x_2-v_d $. By forward integrating of dynamics \eqref{eqn:acc_dynamics} with $ c_2=0 $ (thus taking a linear approximation), the output becomes \cite{ames2020integral_cbf}:
\begin{align}
\hat{y}(t+T)=&-c_1^{-1}\left(c_0-u+m v_d-c_1 e^{-\frac{c_1 T}{m}}\right.\nonumber\\
&\left.\left(x_2(t)+\frac{c_0-u+m v_d}{c_1}\right)\right)
\end{align}

Consequently, the control law $\phi(\bx,\bu)$ is:

\begin{align}
\dot{u}(t)=\alpha c_1\left(e^{\frac{-c_1}{m} T}-1\right)^{-1} \hat{y}(t+T)=: \phi(\bx, u,t) .
\end{align}

Safety constraints for the state follow the "half the speedometer rule", resulting in the CBF 
\begin{align}
h_x(\bx)=x_3-1.8 x_2 \geq 0 \quad \rightarrow \quad \mathcal{S}_x=\left\{\bx:\;h_x(\bx) \geq 0\right\} .\nonumber
\end{align}
The restriction for the input is based on the fact that the force exerted on the wheel must be upper bounded: $ \|u\| \leq m c g $, which gives rise to a DO-ICBF:
\begin{align}
h_u(\bx)=\left(m c g\right)^2-u^2 \geq 0 \quad \rightarrow \quad \mathcal{S}_u=\left\{u:\;h_u(u) \geq 0\right\} .\nonumber
\end{align}
Here, $ c $ represents the $ g $-factor for acceleration and deceleration. The parameters for the simulation are given in Table \ref{table:parameters}.

\subsection{Vehicle model\label{subsec:vehicle_model}}
To model the vehicle, we use the bicycle model which is a simplified representation of a ground vehicle used in control and simulation studies. 
The bicycle model, often used to represent the dynamics of vehicles, captures the essential lateral and longitudinal dynamics by considering the vehicle as a two-wheeled bicycle. 
The dynamics of the vehicle based on the bicycle model are then given by:
    \begin{align}
 \dot{x} = v \cos\psi,\quad\dot{y} = v \sin\psi,\quad\dot{\psi} = \frac{v \tan\delta}{L},\quad\dot{v}=a    
\end{align}
where $(x,y)$ are the coordinates of the center of gravity (CG) of the vehicle, $ \psi $ is its heading angle relative to the inertial frame, $ \delta $ is the steering angle of the front wheel, $ v=0.5\mathrm{m/s} $ is the speed of the vehicle (which is assumed to be constant for the simulations), $a$ is the acceleration and $ L $ is the wheelbase of the vehicle or the distance between the front and rear axles. The state is $\bx:=[x,\;y,\;\psi,\;v]^\mathrm{T}$ and the control input is $u=\delta$ (the acceleration $a$ which is a input is assumed to constant for the simulations). The Stanley controller is a lateral controller, designed to minimize the cross-track error, which is the perpendicular distance from the vehicle to the reference path. 
The steering command $ \delta $ using the Stanley controller is given by:
\begin{align}
    \delta_S = \theta + \arctan \left( \frac{k e}{v} \right) \nonumber 
\end{align}
where $ k>0 $,  $ e $ is the cross-track error and $ \theta $ is the orientation of the path at the closest point to the vehicle. In this case, the integral control law $\phi(\bx,\bu)$ is $\dot{\delta}_S$.

The objective is to track the reference trajectory (denoted by the blue curve in Fig. \ref{fig:stanley_x_and_y}) from the initial state $\bx_0=[15,\;10,\;\pi/2,\;0.5]^\mathrm{T}$ in the presence of an obstacle located at the origin with a radius equal to one. Towards this aim, we define the DO-ICBF as follows:
\begin{align}
b_0(\bx,\bu)=x^2+y^2-1\nonumber
\end{align}
Consequently, the DO-ICBF condition boils down to
\begin{align}
b_1(\bx,\bu):=&\dot{b}_0(\bx,\bu)+\gamma_1(b_0(\bx,\bu))\nonumber\\
=&2x\dot{x}+2y\dot{y}+\gamma_1(b_0(\bx,\bu))\nonumber\\
=&2xv\cos\psi+2yv\sin\psi+\gamma_1(b_0(\bx,\bu))\geq 0\nonumber
\end{align}
Note that $b_1(\bx,\bu)$ does not contain any input term $u$ or integral term $\dot{u}$. Towards that aim, we consider a second order DO-ICBF $b_2(\bx,\bu)$ given by the following
\begin{align}
b_2(\bx,\bu):=&\dot{b}_1(\bx,\bu)+\gamma_2(b_1(\bx,\bu))\nonumber
\end{align}
Note that $\dot{b}_1(\bx,\bu)$ contains the term $\dot{u}$. For the simulations, we consider $\gamma_1(p)=0.2p$, $\gamma_2(p)=p$ and $\gamma_3(p)=p$ for scalar $p$. To ensure that the vehicle does not collide with the obstacle located at the origin we modify the integral term $\dot{u}=\dot{\delta}_S$ to the following
\begin{align}
   \dot{u}=\dot{\delta}_S+v^\star(\bx,u,\hbd) 
\end{align}
where $v^\star$ is obtained by iteratively solving the following QP
\begin{subequations}
\begin{align}
&v^\star(\bx, u,\hbd) = \underset{v \in \mathbb{R}^m}{\operatorname{argmin}}\|v\|^2 \\
\text {subject to } &\frac{\partial b_{2}(\bx,u) }{\partial u}^{\mathrm{T}} v \geq w_2(\bx, u,\hbd)
\end{align}
where $\hbd=0$ in this case. Fig. \ref{fig:stanley_x_and_y} shows the vehicle following (in red) the reference trajectory in blue. Fig. \ref{fig:psi_variaton} shows the variation of the heading angle $\psi$. As shown in Fig. \ref{fig:b_2_variation}, the values of High Order DO-ICBFs $b_1$ and $b_2$ always remains positive with time implying that the set $\mathcal{S}^2:=\mathcal{S}_0\cap\mathcal{S}_1\cap\mathcal{S}_2$ where $\mathcal{S}_0=\{(\bx,\bu):b_0(\bx)\geq 0\}$ , $\mathcal{S}_1=\{(\bx,\bu):b_1(\bx)\geq 0\}$ and $\mathcal{S}_2=\{(\bx,\bu):b_2(\bx,\bu)\geq 0\}$  .
\label{eqn:qp_for_high_order_do_icbf_stanley}
\end{subequations}
\section{Conclusion\label{sec:conclusion}}
In this paper, we presented Disturbance Observer based Integral Control Barrier functions (DO-ICBFs) for safe control synthesis for general nonlinear controlled systems with bounded additive disturbances. Next, we extended the approach to systems with higher relative degree with respect to the auxillary control variable and proposed High Order DO-ICBFs to guarantee safety. One of the interesting directions for future work would be the application of proposed DO-ICBFs to real world applications such as bio-inspired robotic systems such as flapping wing robot in presence of wind gust and collision avoidance for a hypersonic vehicle.
\bibliography{main.bib}

\end{document}